\documentclass[5p]{elsarticle}
\usepackage{mathrsfs}

\usepackage{epsfig}
\usepackage[centertags]{amsmath}
\usepackage{graphics}
\usepackage{subfigure}
\usepackage{booktabs}
\usepackage{amsthm}
\usepackage{amsmath}
\usepackage{amssymb}
\usepackage{psfrag}

\usepackage[ruled]{algorithm2e}
\usepackage{setspace}

\usepackage{url}
\usepackage{psfrag}
\usepackage{xcolor}
\usepackage{amsfonts}

\newtheorem{theorem}{Theorem}

\newcommand{\reals}{{\mathbb{R}}}

\newcommand{\sgn}{{\mathrm{sgn}}}

\DeclareMathOperator*{\argmax}{arg\,max}
\DeclareMathOperator*{\Min}{minimize}
\DeclareMathOperator*{\Max}{maximize}
\newcommand{\st}{\mathrm{s.t.}}

\newcommand{\vx}{{\bf x}}

\newcommand{\ve}{{\bf e}}
\newcommand{\vu}{{\bf u}}
\newcommand{\vv}{{\bf v}}
\newcommand{\vs}{{\bf s}}
\newcommand{\vt}{{\bf t}}
\newcommand{\vz}{{\bf z}}
\newcommand{\vw}{{\bf w}}
\newcommand{\vy}{{\bf y}}

\newcommand{\vU}{{\bf U}}

\begin{document}

\markboth{X. Huang, L. Shi, M. Yan, and J.A.K. Suykens}{Neurocomputing}


\title{Pinball Loss Minimization for One-bit Compressive Sensing: \\ Convex Models and Algorithms}


\author{Xiaolin Huang$^{1,2}$}
\ead{xiaolinhuang@sjtu.edu.cn}
\author{Lei Shi$^{3}$}
\ead{leishi@fudan.edu.cn}
\author{Ming Yan$^{4,}$\corref{cor1}}
\ead{yanm@math.msu.edu}
\author{Johan A.K. Suykens$^{2}$}
\ead{johan.suykens@esat.kuleuven.be}
\address{$^1$Institute of Image Processing and Pattern Recognition,  Shanghai Jiao Tong University, Shanghai 200240, P.R. China}
\address{$^2$KU Leuven, ESAT-STADIUS, Leuven, B-3001 Belgium}
\address{$^3$Shanghai Key Laboratory for Contemporary Applied Mathematics and School of Mathematical Sciences, Fudan University, Shanghai 200433, P.R. China}
\address{$^4$Department of Computational Mathematics, Science and Engineering and Department of Mathematics, Michigan State University, MI 48824, USA}
\cortext[cor1]{Corresponding author.}

\begin{abstract}
The one-bit quantization is implemented by one single comparator that operates at low power and a high rate. Hence \emph{one-bit compressive sensing (1bit-CS)} becomes attractive in signal processing. When measurements are corrupted by noise during signal acquisition and transmission, 1bit-CS is usually modeled as minimizing a loss function with a sparsity constraint. The one-sided $\ell_1$ loss and the linear loss are two popular loss functions for 1bit-CS. To improve the decoding performance on noisy data, we consider the \emph{pinball loss}, which provides a bridge between the one-sided $\ell_1$ loss and the linear loss. Using the pinball loss, two convex models, an elastic-net pinball model and its modification with the $\ell_1$-norm constraint, are proposed. To efficiently solve them, the corresponding dual coordinate ascent algorithms are designed and their convergence is proved. The numerical experiments confirm the effectiveness of the proposed algorithms and the performance of the pinball loss minimization for 1bit-CS.
\end{abstract}

\begin{keyword}
compressive sensing, one-bit, pinball loss, dual coordinate ascent
\end{keyword}

\maketitle

\newpage

\section{Introduction}
Quantization happens in analog-to-digital conversions, and the extreme quantization scheme is to acquire one bit for each measurement.
This scheme only needs a single comparator and has many benefits in hardware implementation such as low power and a high rate.
Suppose we have a linear sensing system $\vu \in \reals^n$ for a signal $\vx \in \reals^n$.
The analog measurement is $\vu^\top \vx$, and the one-bit quantized observation is its sign, i.e., $y = \sgn(\vu^\top \vx)$.
The signal recovery problem related to one-bit measurements can be formulated as finding a signal $\vx$ from the signs of a set of measurements, i.e., $\{\vu_i, y_i\}_{i=1}^m$ with $y_i = \sgn \left( \vu_i^\top \vx  \right)$.

Note that signals with the same direction but different magnitudes have the same one-bit measurements with the same measurement systems, i.e., the magnitude of the signal is lost in this quantization.
Therefore, we have to make an additional assumption on the magnitude of $\vx$.
Without loss of generality, we assume $\|\vx\|_{2} = 1$.
Then the meaning of one-bit signal recovery can be explained as finding the subset of the unit sphere $\|\vx\|_{2} = 1$ partitioned by many hyperplanes.
In general, when the number of hyperplanes becomes larger, the feasible set becomes smaller, and the recovery result becomes more accurate.

However, there may still be infinitely many points in the subset, and we need additional assumptions on the signal to make it unique.
\emph{One-bit compressive sensing (1bit-CS)}, which assumes that the original signal is sparse, is proposed in~\cite{boufounos20081bit} and has attracted much attention in recent years~\cite{laska2011trust,jacques2013robust}.
It tries to recover a sparse signal from the signs of a small number of measurements.
However, different from the regular CS without quantization~\cite{candes2005decoding,donoho2006compressed,eldar2012compressed}, the number of measurements  in 1bit-CS can be larger than the dimension of the signal.
When all the quantized measurements are exact, 1bit-CS algorithms try to find the sparsest solution in the feasible set, i.e.,
\begin{eqnarray}\label{1bit}
\Min_{\vx \in \reals^n} & & \|\vx\|_{0} \nonumber\\
\st && \|\vx\|_{2} = 1,\\
& & y_i =\sgn(\vu_i^\top \vx), ~\forall i = 1, 2, \ldots, m, \nonumber
\end{eqnarray}
where $\|\cdot\|_0$ counts the number of non-zero components.
This problem is difficult to solve due to the $\ell_0$ penalty and the constraint $\|\vx\|_2=1$.
There are several algorithms that approximately solve~\eqref{1bit} or its variants; see~\cite{boufounos20081bit,laska2011trust,boufounos2009greedy,xu2014bayesian}.

In~\eqref{1bit}, we require that $y_i =\sgn(\vu_i^\top \vx)$ holds for all the measurements with the assumption that there is no noise.
However, in real applications, noise is unavoidable in the measurement process, i.e.,
\begin{equation}\label{noise-data}
y_i=\sgn(\vu_i^\top\vx+\varepsilon_i),
\end{equation}
where $\varepsilon_i$ is the noise.
When $\sgn(\vu_i^\top\vx+\varepsilon_i)=\sgn(\vu_i^\top\vx)$ (i.e., $\varepsilon_i$ is small) for all $i$, we can still recover the true signal accurately as in the noiseless case.
However, when the noise $\varepsilon_i$ is large, we may have $\sgn(\vu_i^\top\vx+\varepsilon_i)\neq\sgn(\vu_i^\top\vx)$.
In addition, there could be sign flips on $y_i$ during the transmission.
Note that sign changes because of noise happen with a higher probability when the magnitude of the true analog measurement is small, while sign flips during the transmission happen randomly among the measurements.

With noise or/and sign flips, the feasible set of~\eqref{1bit} excludes the true signal and can become empty.
To deal with noise and sign flips, the constraint $y_i =\sgn(\vu_i^\top \vx)$ is replaced by loss functions to penalize the inconsistency.
The first model is given in~\cite{jacques2013robust}, where the \emph{one-sided $\ell_1$ loss} $\max\{0, -y_i(\vu_i^\top\vx)\}$ is used to measure the sign inconsistency. While~\cite{plan2013robust} considers the \emph{linear loss} $-y_i(\vu_i^\top\vx)$.
Via minimizing the one-sided $\ell_1$ or the linear loss, some robust 1bit-CS models and the corresponding algorithms are proposed in~\cite{jacques2013robust, plan2013robust, yan2012robust,zhang2014efficient}.
These models will be reviewed in Section II.

In this paper, we will consider the trade-off solution between the one-sided $\ell_1$ loss and the linear loss, named \emph{pinball loss}, to establish recovery models for 1bit-CS.
Statistically, the pinball loss is closely related to the concept of quantile; see~\cite{koenker2005quantile,christmann2007svms,steinwart2011estimating} for regression and~\cite{huang2014support} for classification.
Use the following definition for the pinball loss:
\begin{equation}\label{pinball}
\small
L_{\tau,c}(t) = \left\{
\begin{array}{ll}
c+t, & t \geq -c,\\
-\tau (c+t), & t < -c,
\end{array}
\right.
\end{equation}
where $t=-y_i(\vu_i^\top\vx)$. (There is another and equivalent definition of the pinball loss in quantile regression field; see, e.g., \cite{christmann2007svms}.)
It is characterized by parameters $\tau$ and $c$, and it is convex when $\tau\geq -1$.
The one-sided $\ell_1$ loss and the linear loss can be viewed as particular pinball loss functions with $(\tau = 0,~c=0)$ and $(\tau = -1,~c=0)$, respectively.
In other words, $L_{\tau,c}(t)$ provides a bridge from the one-sided $\ell_1$ loss to the linear loss.

In this paper, we will use the pinball loss to establish two convex models to recover signals from 1bit observations.
The first model contains the pinball loss, the $\ell_1$-norm regularization term, and the $\ell_2$-norm ball constraint.
Since both the $\ell_1$-norm and the $\ell_2$-norm are considered, we name it as \emph{Elastic-net Pinball loss model (EPin)}.
For the second model, we put the $\ell_1$-norm term into the constraint and then name it as EPin with sparsity constraint (\emph{EPin-sc}).
To efficiently solve them, the dual problems are derived, and the corresponding dual coordinate ascent algorithms are given.
These algorithms are proved to converge to the optima of the primal problems, and their effectiveness is evaluated on numerical experiments.

This paper is organized as follows.
A brief review on existing 1bit-CS methods is given in Section~\ref{sec:2}.
Section~\ref{sec:4} introduces the pinball loss and then proposes EPin. An efficient algorithm is designed as well.
The discussion on EPin-sc is given in Section~\ref{sec:5}.
The proposed methods are then evaluated on numerical experiments in Section~\ref{sec:6}, showing the performance of the pinball loss in 1bit-CS.
A conclusion is given to end this paper in Section~\ref{sec:7}.

\section{Review on 1bit-CS Models}\label{sec:2}
Let $\vU = [\vu_1, \vu_2, \ldots, \vu_m]$ and $\vy = [y_1, y_2, \ldots, y_m]^\top$ stand for the sensing system and the measurements, respectively. Denote $\vy\circ (\vU^\top \vx)$ as the vector with components $\{y_i(\vu_i^\top\vx)\}$.

In order to efficiently recover the sparse signal in 1bit-CS, the $\ell_0$ penalty is replaced by the $\ell_1$ norm as in regular compressive sensing~\cite{boufounos20081bit,laska2011trust}.
In order to pursue the convexity, the non-convex sphere constraint $\|\vx\|_{2} = 1$ is replaced by a convex constraint in~\cite{plan2013one}, and a convex model is established as follows:
\begin{eqnarray}\label{1bit-l1-convex}
\Min_{\vx \in \reals^n} & & \|\vx\|_{1} \\
\st & &  \|\vU^\top \vx\|_{1} = \beta, ~ \vy\circ (\vU^\top\vx)\geq 0 \nonumber,
\end{eqnarray}
where $\beta$ is a given positive constant.
Note that~\eqref{1bit-l1-convex} can be reformulated as a linear programming problem because the first constraint $\|\vU^\top \vx\|_{1}=\beta$ becomes $\sum_{i=1}^m y_i (\vu_i^\top \vx)=\beta$ if the second constraint is satisfied.
However, its solution is not necessarily located on the unit sphere. Hence one needs to project the solution onto the unit sphere, and the projected solution is independent of $\beta$.

As we mentioned before, the constraint $\vy\circ (\vU^\top\vx)\geq 0$ assumes the noiseless case, i.e., there is no sign changes in $\vy$. To deal with noise and sign flips, one replaces the constraint $\vy\circ (\vU^\top \vx) \geq 0$ by a loss function.
Using the one-sided $\ell_1$ loss,~\cite{jacques2013robust} introduces the following robust model:
\begin{eqnarray}\label{BIHT1}
\Min_{\vx\in\reals^n} & & \frac{1}{m} \sum_{i=1}^m L_{0,0} (-y_i (\vu_i^\top \vx)) \\
\st & & \|\vx\|_{0} = K, \|\vx\|_{2} = 1, \nonumber
\end{eqnarray}
where $K$ is the number of non-zero components in the true signal.
Then \emph{Binary Iterative Hard Thresholding with a one-sided $\ell_1$-norm} (BIHT) is proposed to solve it approximately.
Modifications of BIHT for sign flips are designed in~\cite{yan2012robust} to improve its robustness to sign flips.
There are also several ways to deal with sign changes because of noise:~\cite{bahmani2013robust} uses maximum likelihood estimation;~\cite{Fang2014201} uses a logistic function;~\cite{dai2014noisy} uses a robust one-sided $\ell_0$ penalty.

Note problem~\eqref{BIHT1} is non-convex, and BIHT only approximately solves it.
To get a convex model, the unit sphere constraint $\|\vx\|_{2} = 1$ is relaxed to the unit ball constraint $\|\vx\|_2\leq 1$, and the sparsity constraint $\|\vx\|_0 =  K$ is replaced by an $\ell_1$ constraint $\|\vx\|_1\leq s$.
Moreover, the one-sided $\ell_1$ loss is replaced by a linear loss to avoid the trivial zero solution, and minimizing the linear loss can be explained as maximizing the correlation between $y_i$ and $\vu_i^\top \vx$.
With those modifications,~\cite{plan2013robust} gives the following convex model for robust 1bit-CS:
\begin{eqnarray}\label{Plan}
\Min_{\vx\in\reals^n} & & \frac{1}{m} \sum_{i=1}^m\nolimits L_{-1,0}(-y_i (\vu_i^\top \vx)) \\
\st & &  \|\vx\|_{1} \leq s, ~\|\vx\|_{2} \leq 1, \nonumber
\end{eqnarray}
where $s$ is a given positive constant.

One can also put the $\ell_1$-norm in the objective function.
The corresponding problem is given in~\cite{zhang2014efficient}:
\begin{eqnarray}\label{Zhang}
\Min_{\vx\in\reals^n} & & \mu \|\vx\|_{1} + \frac{1}{m} \sum_{i=1}^m L_{-1,0}(-y_i (\vu_i^\top \vx)) \nonumber \\
\st & &  \|\vx\|_{2} \leq 1,
\end{eqnarray}
where $\mu$ is the regularization parameter for the $\ell_1$-norm.
In the rest of this paper, we call~\eqref{Plan} \emph{Plan's model} and~\eqref{Zhang} \emph{the passive model}.
Both problems~\eqref{Plan} and~\eqref{Zhang} are convex, and there is a closed-form solution for~\eqref{Zhang}.

{Similar to regular compressive sensing, suitable nonconvex penalties can be used in (\ref{Plan}) or (\ref{Zhang}) to replace the $\ell_1$-norm to enhance the sparsity. For example, smoothly clipped absolute deviation \cite{fan2001variable} and minimax concave penalty \cite{zhang2010nearly} are discussed in~\cite{zhu2015towards} for 1bit-CS. In addition, fast algorithms with analytical solutions for positive homogeneous penalties is recently given by \cite{huang2018nonconvex}. The use of nonconvex penalties can enhance the sparsity and has shown promising performance when there are only a few measurements. However, nonconvex penalties for 1bit-CS are currently restricted to linear loss, due to the computational effectiveness.}

\section{Pinball Loss Minimization with Elastic-net} \label{sec:4}
\subsection{Pinball loss and EPin}
In robust 1bit-CS models, the loss function plays an important role. Intuitively, the loss function can be explained as a penalty given to the inconsistency of $y_i$ and $\sgn(\vu_i^\top \vx)$. Plan's model, the passive model, and BIHT have the same loss when $y_i \neq \sgn(\vu_i^\top \vx)$, but there is a big difference for a measurement that has a correct sign, i.e., $y_i (\vu_i^\top \vx) > 0$. In that case, BIHT, which applies the $\ell_1$-sided loss, does not give any penalty but Plan's model and the passive model, which use the linear loss, give a gain (negative penalty) to encourage a larger $y_i (\vu_i^\top \vx)$.

In this paper, we consider the trade-off between the linear loss and the $\ell_1$-sided loss. Specifically, when $y_i (\vu_i^\top \vx)$ is negative, we give a penalty as the existing losses and when $y_i (\vu_i^\top \vx)$ is large enough, we still give a gain but with a relatively small weight. Mathematically, this kind of loss is formulated as the pinball loss defined in (\ref{pinball}). The parameter $|\tau|$ describes the ratio of the weights for $y_i (\vu_i^\top \vx) > c$ and $y_i (\vu_i^\top \vx) \leq c$. The one-sided $\ell_1$-norm does not care about the samples with the correct signs, then $\tau=0$; the linear loss gives the equal emphasis on all the samples, thus, $\tau=-1$. Note that we have an additional parameter $c$: the changing point for the large and the small penalty.

Applying the pinball loss in 1bit-CS, we propose the following model:
\begin{eqnarray}\label{elastic-net-primal}
 \min\limits_{\vx} & &  P(\vx) \triangleq \mu \|\vx\|_{1} + \frac{1}{m} \sum_{i=1}^m L_{\tau,c}(- y_i (\vu_i^\top \vx)) \nonumber \\
 \st & &  \|\vx\|_{2} \leq 1.
\end{eqnarray}
Here the parameter $\mu$ is used to balance the regularization and the loss terms.
We name~\eqref{elastic-net-primal} \emph{Elastic-net Pinball loss model (EPin)} because it involves both the $\ell_1$ and the $\ell_2$-norms.
When $\tau = -1$, the pinball loss becomes the linear loss, and EPin reduces to the passive model~\eqref{Zhang}, for which there is a closed-form solution.
When $\tau>-1$, analytic solutions are not available, and we will introduce its dual problem and then a dual coordinate ascent method.

Before discussing the dual problem and the algorithm, we here numerically show the performance of the pinball loss minimization. The underlying signal, denoted by $\bar \vx$, has $n$ components with $K$ non-zero ones. Non-zero components are first generated following the standard Gaussian distribution, and then are normalized such that $\|\bar \vx\|_{2} = 1$. We take $m$ binary observations with measurement vector $\vu$ drawn from the standard Gaussian distribution. Throughout the numerical experiments, we use Gaussian noise and the noise level is measured by the ratio of the variance of $\varepsilon$ to that of $\vu^\top \bar \vx$, denoted by $s_n$. Moreover, there could be sign flips, of which the ratio is denoted by $r_f$. Suppose that the recovered signal is $\tilde \vx$, and then the Signal-to-Noise Ratio (SNR) in dB, defined below,
\begin{equation}\label{snr}
\mathrm{SNR}_\mathrm{dB}(\bar \vx, \tilde \vx) = 10 \log_{10} \left( {\|\bar \vx\|^2_{2}}\big /{\|\bar \vx - \tilde \vx\|^2_{2}} \right),
\end{equation}
is used to measure the recovery quality.

To investigate the role of the bias term $c$, we choose $r_f = 10\%$ and $s_n = 10$, but vary $c$ from $0$ to $1.5$. First, we choose $\tau = 0$.
The average SNR over $200$ trials is plotted in Fig.\ref{fig-example-2:a}.
This experiment shows the importance of using a non-zero $c$ for $\tau = 0$. Simply minimizing the one-sided $\ell_1$ loss has no capability to recover the signal for small $c$, and a non-convex constraint is needed, like $\|\vx\|_2 = 1$ used in (\ref{BIHT1}).
In Fig.\ref{fig-example-2:b}, we display the performance for different $c$ values when $\tau = -0.5$. The two figures imply that the performance with a large $c$ is similar. Especially, with further tuning $\mu$, there is little difference for different $c$ values when $c$ is large enough. In the rest, we choose $c = 1$. Another important parameter is $\mu$, which is suggested in \cite{zhang2014efficient} to be $\sqrt{\log(n)/m}$ when $\tau = -1$. For other $\tau$ values, this setting is not necessarily optimal but it at least implies a reasonable range. In this paper, we will use cross-validation to tune it around $\sqrt{\log(n)/m}$.

\begin{figure}[htbp]
  \centering
  \subfigure[]{
    \label{fig-example-2:a} 
    \includegraphics[width=0.48\linewidth]{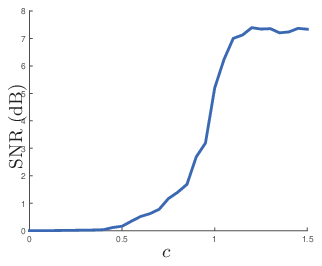}}
  \subfigure[]{
    \label{fig-example-2:b} 
    \includegraphics[width=0.48\linewidth]{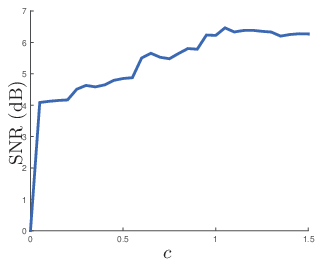}}
  \caption{Average SNR of EPin for different $c$ values with $m = 500, n = 1000$. In this experiment, $\mu = \sqrt{\log(n)/m}$ and the observations are corrupted by Gaussian noise with $s_n = 10$ and sign flips with $r_f = 10\%$. (a) $\tau = 0$ (this also could be regarded as a modification from the passive model with an additional bias); (b) $\tau = -0.5$.}\label{fig-example-2}
\end{figure}

In Fig.\ref{fig-tau-mu}, the average SNR for different $\tau$ and $\mu$ is displayed. As mentioned previously, $\tau = -1$ corresponds to the linear loss employed in the passive model, for which  $\mu = \sqrt{\log(n)/m}$ is suggested by \cite{zhang2014efficient}. The results imply that suitably selecting  $\tau$ and $\mu$ can improve the recovery performance by about 2dB for this case. The improvement amplitude depends on the number of measurements, the sparsity level, and the noise level.

\begin{figure}[htbp]
  \centering
  \subfigure[]{
    \label{fig-tau-mu:a} 
    \includegraphics[width=0.48\linewidth]{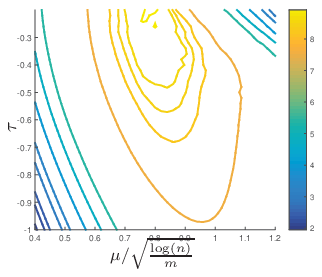}}
  \subfigure[]{
    \label{fig-tau-mu:b} 
    \includegraphics[width=0.48\linewidth]{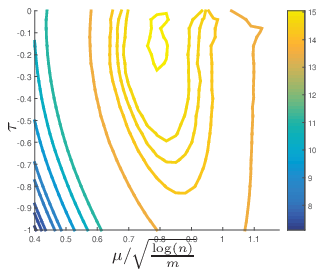}}
  \caption{Average SNR of EPin for different $\tau$ and $\mu$. In this experiment, $n = 1000, K = 10$ and the observations are corrupted by Gaussian noise with $s_n = 10$ and sign flips with $r_f = 10\%$. (a) $m = 500$; (b) $m = 2000$.}\label{fig-tau-mu}
\end{figure}

\subsection{Dual problem}
In order to obtain the dual problem of Epin, we reformulate~\eqref{elastic-net-primal} as:
\begin{equation}\label{elastic-net-primal-1}
\begin{aligned}
\Min_{\vx,\ve, \vz}\quad  & \textstyle \mu \|\ve\|_1 + \frac{1}{m} \sum_{i=1}^m L_{\tau,c}(z_i)+\iota_2(\vx)  \\
\st \quad  & \vx=\ve, ~~  - \vy\circ(\vU^\top\vx) = \vz,
\end{aligned}
\end{equation}
where $\iota_2(\vx)$ has value 0 if $\|\vx\|_2\leq 1$ and $+\infty$ otherwise.
Let $\vs\in\reals^n$ and $\vt\in\reals^m$.
Then the corresponding Lagrangian function is
\begin{align*}
 { \cal{L}}(\vx,\ve, \vz, \vs,\vt)
 ~ = ~ & \textstyle\mu \|\ve\|_1 + \frac{1}{m}\sum_{i=1}^m L_{\tau,c}(z_i) + \iota_2(\vx) \\
    & + \vs^\top(\vx-\ve)+\vt^\top ( - \vy\circ(\vU^\top\vx)-\vz) .
\end{align*}
Minimizing over primal variables $\vx, \ve, \vz$, we have:
\begin{align*}
& \min_{\vx} \; \iota_2(\vx) + \vs^\top \vx - \vt^\top(\vy\circ (\vU^\top\vx)) \\
&  ~~~~~~~~~~~~~~~~~~~~  =   \textstyle -\left\| \sum_{i=1}^m t_i y_i \vu_i - \vs \right\|_{2},
\end{align*}
\begin{align*}
& \min_{\ve} \; \mu \|\ve\|_1 - \vs^\top \ve  =  \; \textstyle \left \{
\begin{array}{ll}
0,       & \mathrm{if}~~ \|\vs\|_{\infty} \leq \mu, \\
-\infty, & \mathrm{otherwise},
\end{array}
\right.\\
& \min_{z_i} \; \frac{1}{m} L_{\tau,c}(z_i) - t_iz_i  =  \; \textstyle \left \{ {
\begin{array}{ll}
ct_i,       & \mathrm{if}~~ -\frac{\tau}{m} \leq t_i \leq \frac{1}{m}, \\
-\infty, & \mathrm{otherwise}.
\end{array}
} \right.
\end{align*}
The dual problem of~\eqref{elastic-net-primal-1}, i.e., $\max_{\vs, \vt}\limits \min_{\vx, \ve, \vz}\limits { \cal{L}}(\vx,\ve, \vz, \vs,\vt)$, is
\begin{equation}
\label{elastic-dual}
\begin{aligned}
\Max_{\vs, \vt}\quad &\textstyle   D(\vs, \vt) \triangleq c\sum_{i=1}^m t_i  - \left\| \sum_{i=1}^m t_i y_i \vu_i - \vs \right\|_{2}    \\
\st\quad &  \|\vs\|_\infty \leq \mu, ~~  -\frac{\tau}{m} \leq \vt \leq \frac{1}{m}.
\end{aligned}
\end{equation}
From the optimal dual variables $\vs^*, \vt^*$, we can easily find an optimal $\vx^*$ for (\ref{elastic-net-primal}):
\begin{enumerate}
\item If $\sum_{i=1}^m t^*_i y_i \vu_i - \vs^* \neq {\bf 0}$, the optimal $\vx^*$ can be obtained as
\begin{equation*}\label{primal-dual}
\vx^* = \textstyle\left( \sum_{i=1}^m t^*_i y_i \vu_i - \vs^* \right)/\left\| \sum_{i=1}^m t^*_i y_i \vu_i - \vs^* \right\|_{2}.
\end{equation*}
\item If $\sum_{i=1}^m t^*_i y_i \vu_i - \vs^* = {\bf 0}$, the optimal $\vx^*$ is not necessarily unique, and any $\vx^*$ satisfying the conditions below is optimal.
\begin{subequations}\label{sol_con}
\small
\begin{align}
\|\vx^*\|_{2}\leq 1,~~&\label{sol_cons_a}\\
x^*_j = 0, ~~& \mbox{ if } |s^*_j|<\mu,\label{sol_cons_b}\\
x^*_j \geq 0, ~~ & \mbox{ if } s^*_j =\mu,\\
x^*_j \leq 0, ~~ & \mbox{ if } s^*_j = -\mu,\label{sol_cons_d}\\
c  - y_i (\vu_i^\top \vx^*) \geq 0, ~~ & \mbox{ if } t^*_i={1/m}, \label{sol_cons_e}\\
c  - y_i (\vu_i^\top \vx^*) \leq 0, ~~ & \mbox{ if } t^*_i={-\tau/m},\label{sol_cons_f}\\
c  - y_i (\vu_i^\top \vx^*) = 0, ~~ & \mbox{ if } t^*_i\in ({-\tau/m},1/m). \label{sol_cons_g}
\end{align}
\end{subequations}
\end{enumerate}
\noindent {\bf Remark:} When $\tau=-1$, any $\vx^*$ satisfying~\eqref{sol_cons_a}-\eqref{sol_cons_d} is optimal. This generalizes the result for the passive model~\cite[Lemma 1]{zhang2014efficient}.

Let us define two hypercubes for $\vz\in\reals^n$:
\begin{eqnarray*}\textstyle
& & {\cal A} = \left\{\vz = \sum_{i=1}^m t_i y_i \vu_i \Big | -\frac{\tau}{m} \leq \vt \leq \frac{1}{m}\right\}, \\
& & {\cal B} = \left\{\vz \Big | -\mu \leq \vz \leq \mu\right\}.
\end{eqnarray*}
If ${\cal A} \bigcap {\cal B} = \emptyset$, then the optimal $\vx^*$ will always be on the unit sphere. The case for ${\cal A}\cap{\cal B}\neq\emptyset$ is more complicated:
If $c=0$, the optimal dual objective is $0$, and the primal objective becomes zero when $\vx=\bf 0$, so $\bf 0$ is optimal to the primal problem~\cite{zhang2014efficient}. However, if $c>0$, we may still have $\left\|\sum_{i=1}^m t^*_i y_i \vu_i\right\|_\infty>\mu$ and $\vx^*$ is still on the unit sphere.

In order to get the optimal $\vx^*$ on the unit sphere, we can choose a small $\mu$ because a smaller $\mu$ leads to a smaller $\cal{B}$, which then can lead to an empty $\cal{A}\cap\cal{B}$.

\subsection{Dual coordinate ascent algorithm}
The motivation of solving EPin from the dual space instead of directly solving (\ref{elastic-net-primal}) is that the constraints in (\ref{elastic-dual}) are not coupled, which allows us to design a coordinate update algorithm. The subproblems of dual variables are:

1) $s_j$-subproblem: $D(\vs,\vt)$ is separable with respect to $\vs$, and $s_j$ can be computed in parallel via
\begin{equation}\label{beta-update-solution}
\textstyle s_j = \max \left\{ -\mu, \min\left\{\mu, \left(\sum_{i=1}^m t_i y_i \vu_i\right)_j \right\}\right\}.
\end{equation}

2) $t_i$-subproblem: Let us consider updating $t_i$ to $t_i + d_i$. It is a univariate optimization problem on $d_i$:
\begin{equation}\label{di_problem}
\textstyle \Max\limits_{-\frac{\tau}{m} \leq t_i + d_i \leq \frac{1}{m}} c d_i  - \left\| y_i \vu_i d_i + \sum_{i=1}^m t_i y_i \vu_i - \vs \right\|_{2}.
\end{equation}
Denote $\vw = \sum_{i=1}^m t_i y_i \vu_i - \vs$. Problem~\eqref{di_problem} becomes
\begin{equation*}
\Max_{-\frac{\tau}{m} \leq t_i + d_i \leq \frac{1}{m}} c d_i  - \sqrt{\|\vu_i\|_2^2d^2_i+2  y_i\vu_i^\top\vw d_i +\|\vw\|_2^2},
\end{equation*}
and its optimal solution $d_i^*$ can be calculated as follows:
\begin{itemize}
\item If $\|\vu_i\|_2\leq c$, the objective function is non-decreasing. We have that $d_i^* = 1/m-t_i$ is optimal and update $t_i$ to be $1/m$.
\item If $\|\vu_i\|_2>c$, we define $ a_d  =  \|\vu_i\|^2_{2}(\|\vu_i\|^2_{2} - c^2)$, $b_d  =  2 (\|\vu_i\|^2_{2} - c^2) y_i\vu_i^\top \vw$, $c_d  =  (\vu_i^\top \vw)^2 - c^2\|\vw\|^2_{2}$, then there is
\begin{equation}\label{xi-update-solution}
\textstyle d_i^* = \max\left\{ -\frac{\tau}{m}- t_i, \min\left\{ \frac{1}{m} - t_i, \bar d_i\right\} \right\},
\end{equation}
where
\begin{equation*}
\textstyle\bar d_i = \frac{-b_d + \sqrt{b_d^2 - 4a_dc_d}}{2a_d}.
\end{equation*}
\end{itemize}

Summarizing the previous discussion, we give the dual coordinate ascent method for~\eqref{elastic-net-primal} in Alg.\ref{EPin-algorithm}, which is fast because each subproblem has an analytical solution. Moreover, the next theorem states that its output is optimal.
\begin{algorithm}[h]
\caption{Dual coordinate ascent for EPin}\label{EPin-algorithm}
\Indp   Set $l := 0,\vs^0 := {\bf 0}_{n \times 1}, \vt^0 := -\frac{\tau}{m}{\bf 1}_{m \times 1}$; \\
Calculate $\vw := \sum_{i=1}^m t^0_i y_i \vu_i - \vs^0$;\\
\Repeat{$\vt^{l} = \vt^{l-1}$}
{
\For{ $i = 1, 2, \ldots, m$}
{

\eIf{$c \geq \|\vu_i\|_{2}$}
{
 $d^*_i := \frac{1}{m} - t_i^l$;
}
{
 Calculate  $d^*_i$ by (\ref{xi-update-solution});
}
 $\vw := \vw + y_i \vu_i d^*_i$;

 $t^{l+1}_i := t_i^l + d^*_i$;
}
 Calculate $\vs^{l+1}$ by (\ref{beta-update-solution}) and update $\vw := \vw + \vs^l - \vs^{l+1}$;  $l: = l + 1$;\\
}
\eIf{$\|\vw\|_2 >0$}
{
$\vx := {\vw\over\|\vw\|_2}$;
}
{
Find $\vx$ that satisfies~\eqref{sol_con};
}
\end{algorithm}

\begin{theorem}\label{thm-2}
The dual coordinate ascent for EPin (Alg。~\ref{EPin-algorithm}) converges to an optimal solution of~\eqref{elastic-net-primal}.
\end{theorem}

\begin{proof} Suppose that $\vx^*$ is the output of Alg.~\ref{EPin-algorithm} and $\vs^*, \vt^*$ are the corresponding coordinate optimum for (\ref{elastic-dual}). We are going to prove that $\vx^*$ is optimal to (\ref{elastic-net-primal}). This proof considers two different cases:

{\bf Case 1} ($\vw\neq \bf 0$): We have $\|\vx^*\|_2=1$ and the algorithm shows that $\{s_j^*\}$ and $\{t_i^*\}$ are coordinate maxima of~\eqref{elastic-dual}. Consider a small change on $t_i$, denoted by $\Delta t_i$, and define the following function
\begin{eqnarray*}
\textstyle h(\Delta t_i) & \triangleq & c \Delta t_i  - \left\| y_i \vu_i \Delta t_i + \vw \right\|_{2},
\end{eqnarray*}
of which the gradient at $\Delta t_i=0$ is
\[
\left.\frac{d h(\Delta t_i)}{d \Delta t_i} \right|_{\Delta t_i=0} = c - \frac{ y_i \vu_i^\top \vw}{\|\vw\|_{2}} = c - y_i(\vu_i^\top \vx^*).
\]
Since $\vt^*$ is the coordinate optimum, $\Delta t_i = 0$ is the maximum of $h(\Delta t_i)$ under the condition that $-\frac{\tau}{m} \leq t^*_i + \Delta t_i \leq \frac{1}{m}$. Thus,
\begin{itemize}
\item if $t^*_i=1/m$, then $y_i(\vu_i^\top\vx^*)\leq c$;
\item if $t^*_i=-\tau/m$, then $y_i(\vu_i^\top\vx^*)\geq c$;
\item if $t^*_i\in(-\tau/m,1/m)$, then $y_i(\vu_i^\top\vx^*)= c$.
\end{itemize}
In other words,
\begin{equation}\label{subgradient-p-2}
-\sum_{i=1}^m t^*_i y_i \vu_i \in \frac{\partial {1\over m}\sum_{i=1}^m L_{\tau,c} ( - y_i (\vu^\top_i \vx)) }{\partial \vx}\Big |_{\vx = \vx^*}.
\end{equation}

From the calculation of $\vs^*$ (c.f.~\eqref{beta-update-solution}), we have:
\begin{itemize}
\item if $-\mu < s^*_j < \mu$, then $w_j = \left(\sum_{i=1}^m t^*_i y_i \vu_i\right)_j - s^*_j = 0$, i.e, $x_j^*= 0$;
\item if $s^*_j = \mu$, then $x_j^*\geq0$;
\item if $s^*_j = -\mu$, then $x_j^*\leq0$;
\end{itemize}
which means that $\vs^* \in \frac{\partial \mu \|\vx\|_{1}}{\partial \vx} \big |_{\vx = \vx^*}$.
Together with (\ref{subgradient-p-2}), we have
\[
\vs^* - \sum_{i=1}^m\nolimits t^*_i y_i \vu_i \in \frac{\partial P(\vx)}{\partial \vx} \Big |_{\vx = \vx^*},
\]
from which it follows that
\[
\vx^* = \frac{\sum_{i=1}^m t^*_i y_i \vu_i - \vs^*}{\|\sum_{i=1}^m t^*_i y_i \vu_i - \vs^*\|_{2}}
\]
is optimal to (\ref{elastic-net-primal}).

{\bf Case 2} ($\vw= \bf 0$): in this case, $\vx^*$ satisfies \eqref{sol_con}, then
\begin{eqnarray*}
P(\vx^*) & = & \textstyle \mu\|\vx^*\|_{1} + \sum_{i=1}^m t^*_i (c - y_i(\vu^\top_i \vx^*)) \\
         & = & \textstyle \mu\|\vx^*\|_{1} - \sum_{i=1}^m t^*_i y_i(\vu^\top_i \vx^*) + c \sum_{i=1}^m t^*_i.
\end{eqnarray*}
Note that $\vw = \sum_{i=1}^m t^*_i y_i\vu_i - \vs^* = {\bf 0}$, we have
\begin{eqnarray*}
\textstyle \sum_{i=1}^m t^*_i y_i(\vu^\top_i \vx^*) & = & \textstyle \left(\sum_{i=1}^m\nolimits t^*_i y_i \vu_i\right)^\top \vx^* 
\end{eqnarray*}
\begin{eqnarray*}
~~~~~~~~~~~~~~~~~~~~~~~~~ & = & \textstyle (\vs^*)^\top \vx^* =\mu\|\vx^*\|_{1},
\end{eqnarray*}
where the last equality comes from \eqref{sol_cons_b}-\eqref{sol_cons_d}. Therefore, we have that
\[
\textstyle P(\vx^*) =  c \sum_{i=1}^m t^*_i = D(\vs^*, \vt^*),
\]
i.e., the duality gap is zero and $\vx^*$ is optimal to (\ref{elastic-net-primal}).
\end{proof}

\noindent {\bf Remark 3:} Both Alg.~\ref{EPin-algorithm} and the proof of Theorem~\ref{thm-2} suggest that if $c\geq \|\vu_i\|_2$ for all $i$, then $t_i^*=1/m$, and EPin reduces to the passive model no matter what $\tau$ is. It happens because $y_i(\vu_i^\top\vx)\leq c$ for all $\vx$ in the $\ell_2$-norm ball. Thus, we choose $c$ to be much smaller than most $\|\vu_i\|_2$.

In practice, we can set a maximum number of iterations $l_{\max}$ and use $\|\vt^{l} - \vt^{l-1}\|_{\infty} < \delta$ as the stopping criterion. Here $\delta$ is a small positive number. In the following experiments, we set $l_{\max} = 500$ and $\delta={(1+\tau)/(100m)}$.

\section{EPin with Sparsity Constraint}\label{sec:5}

In the previous section, we considered the pinball loss minimization with the $\ell_1$-norm regularization and the $\ell_2$-norm constraint. Similarly to Plan's model~\eqref{Plan}, we can put the $\ell_1$-norm term in the constraint when there is prior-knowledge about the $\ell_1$-norm of the true signal. Specifically, the new model is
\begin{equation}\label{EPin-sc-primal}
\begin{aligned}
\Min_{\vx\in\reals^n} \quad &  \frac{1}{m} \sum_{i=1}^m L_{\tau,c}(- y_i (\vu_i^\top \vx))  \\
\st \quad & \|\vx\|_{1} \leq \alpha, ~~~ \|\vx\|_{2} \leq 1,
\end{aligned}
\end{equation}
which is named an \emph{Elastic-net Pinball loss  with sparsity constraint (EPin-sc)}.

When $\tau = -1$, EPin-sc reduces to Plan's model (\ref{Plan}). For Plan's model, there is no efficient algorithm until now and, CVX, one standard convex optimization toolbox \cite{grant2008cvx}, was suggested in \cite{zhang2014efficient} to solve it. In the following, we will establish a dual coordinate ascent algorithm to solve (\ref{EPin-sc-primal}), and this method is also applicable to Plan's model.

To derive the dual problem, we reformulate (\ref{EPin-sc-primal}) as
\begin{equation*}
\begin{aligned}
\Min_{\vx, \ve, \vz} \quad &  \iota_1(\ve) + \frac{1}{m} \sum_{i=1}^m L_{\tau,c}(z_i) + \iota_2(\vx)\\
\st \quad & \vx = \ve, ~~ -\vy\circ(\vU^\top \vx) = \vz,
\end{aligned}
\end{equation*}
where $\iota_1(\ve)$ returns 0 if $\|\ve\|_1\leq \alpha$ and $+\infty$ otherwise.
Then the corresponding Lagrangian function is:
\begin{align*}
 { \cal{L}}(\vx,\ve, \vz, \vs,\vt)
 ~ = ~ & \; \iota_1(\ve) + \frac{1}{m}\sum_{i=1}^m L_{\tau,c}(z_i) + \iota_2(\vx) \\
   & +  \vs^\top(\vx-\ve)+\vt^\top (- \vy\circ (\vU^\top \vx) -\vz) .
\end{align*}
Therefore, the dual problem of (\ref{EPin-sc-primal}) can be derived in the same way as in the previous section:

\begin{equation}
\label{EPin-sc-dual}
\begin{aligned}
\Max_{\vs, \vt}\quad &  \textstyle c\sum_{i=1}^m t_i -\alpha \|\vs\|_{\infty} -  \left\| \sum_{i=1}^m t_i y_i \vu_i - \vs \right\|_{2}    \\
\st\quad &  -\frac{\tau}{m} \leq \vt \leq \frac{1}{m}.
\end{aligned}
\end{equation}
After obtaining the optimal dual variables $\vs^*$ and $\vt^*$, the optimal $\vx^*$ to~\eqref{EPin-sc-primal} can be constructed as follows,
\begin{enumerate}
\item If $\sum_{i=1}^m t^*_i y_i \vu_i - \vs^* \neq {\bf 0}$, the optimal $\vx^*$ is
\begin{equation*}
\textstyle \vx^* = \left( \sum_{i=1}^m t^*_i y_i \vu_i - \vs^* \right){/}\left\| \sum_{i=1}^m t^*_i y_i \vu_i - \vs^* \right\|_{2}.
\end{equation*}

\item If $\sum_{i=1}^m t^*_i y_i \vu_i - \vs^* = {\bf 0}$, the optimal $\vx^*$ is not necessarily unique, and all $\vx^*$ satisfying conditions below are optimal.
\begin{subequations}\label{sol_cons_2}
\small
\begin{align}
\|\vx^*\|_{2}\leq 1,~~&\label{sol_cons_2_a}\\
\|\vx^*\|_{1}\leq s, ~~& \label{sol_cons_2_b}\\
{\vs^*}^\top \vx^* = \alpha \|\vs^*\|_{\infty},~~&\label{sol_cons_2_c}\\
c  - y_i (\vu_i^\top \vx^*) \geq 0, ~~ & \mbox{ if } t^*_i={1/m}, \label{sol_cons_2_d}\\
c  - y_i (\vu_i^\top \vx^*) \leq 0, ~~ & \mbox{ if } t^*_i={-\tau/m},\label{sol_cons_2_e}\\
c  - y_i (\vu_i^\top \vx^*) = 0, ~~ & \mbox{ if } t^*_i\in ({-\tau/m},1/m). \label{sol_cons_2_f}
\end{align}
\end{subequations}
\end{enumerate}

Same as in the previous section, we can update $t_i$ and $\vs$ in turn to efficiently solve (\ref{EPin-sc-dual}). Minimization on $t_i$ is the same as for EPin, i.e., $t^{l+1}_i = t^l_i + d^*_i$, where $d^*_i$ is computed by (\ref{xi-update-solution}).

However, the subproblem on $\vs$, i.e.,
\begin{align}\label{beta-update-solution-2}
\textstyle\Max\limits_\vs \quad -\alpha \|\vs\|_{\infty} - \left\| \sum_{i=1}^m t_i y_i \vu_i - \vs \right\|_{2},
\end{align}
is no longer separable. (\ref{beta-update-solution-2}) can be equivalently written as
\begin{align}\label{beta-update-solution-3}
\textstyle\Min\limits_{\xi,~\vs} ~~ \alpha \xi + \sqrt{\sum_{i=1}^m (v_i - s_i)^2},~~ \mathrm{s.t.}  ~~  |s_i| \leq \xi, \forall i,
\end{align}
where $\vv = \sum_{i=1}^m t_i y_i \vu_i$.
Fix $\xi$, and problem~\eqref{beta-update-solution-3} becomes
\begin{eqnarray*}
\Min\limits_{\vs} &   \sqrt{\sum_{i=1}^m (v_i - s_i)^2},~~~ \mathrm{s.t.}  ~~  |s_i| \leq \xi, \forall i,
\end{eqnarray*}
of which the optimal solution is
\begin{equation}\label{beta-t-v}
\small
s_i = B_{v_i}(\xi) \triangleq \left\{
\begin{array}{ll}
\sgn(v_i)\xi, & |v_i| > \xi,\\
v_i,        & |v_i| \leq \xi.
\end{array}
\right.
\end{equation}
Plugging~\eqref{beta-t-v} into~\eqref{beta-update-solution-3}, we have a problem of $\xi$,
\begin{equation}\label{t-subproblem}
\textstyle\Min\limits_{\xi \geq 0} \quad T(\xi) \triangleq \alpha \xi + \sqrt{\sum_{|v_i| > \xi} (|v_i| - \xi)^2}.
\end{equation}
This is a convex univariate problem, and its optimizer $\xi^*$ either equals to zero or satisfies the first-order optimality condition $T'(\xi^*) = 0$, where
\[
T'(\xi) = \alpha - \frac{\sum_{|v_i| > \xi} (|v_i| - \xi)}{\sqrt{\sum_{|v_i| > \xi} (|v_i| - \xi)^2}}.
\]
Note that $T'(\xi)$ is a piecewise smooth function, of which the segment is given by $[|v_{[k+1]}|, |v_{[k]}|]$. Here, $v_{[k]}$ stands for the $k$-th component of $\vv$ in the order of the absolute value, i.e., $|v_{[n]}| \leq\cdots\leq |v_{[1]}|$. Moreover, $T'(t)$ is an increasing function. So it is easy to find the segment containing the solution of  $T'(\xi) = 0$. Specifically, we select $k^*$ such that
\begin{equation}\label{k-piecewise}
T'\left( |v_{[k^*+1]}|\right) \leq 0 ~~\mathrm{and}~~ T'\left( |v_{[k^* ]} |\right) > 0.
\end{equation}
Then $\xi^*$ is in $\left[|v_{[k^* + 1]}|, |v_{[k^*]}|\right)$, from which it follows that it is the solution to the following quadratic equation:
\begin{eqnarray*}
& & \textstyle (k^* - \alpha^2)k^* \xi^2 - 2(k^* - \alpha^2)\left(\sum_{k = 1}^{k^*}\nolimits |v_{[k]}|\right) \xi \\
  & & \textstyle ~~~~ + \left(\sum_{k = 1}^{k^*}\nolimits |v_{[k]}|\right)^2 - \alpha^2 \left(\sum_{k = 1}^{k^*}\nolimits |v_{[k]}|^2\right) = 0.
\end{eqnarray*}
Thus, the optimizer for (\ref{t-subproblem}) is analytically given by
\begin{equation}\label{optimal_t}
\xi^* = \frac{-b_\xi - \sqrt{b_\xi^2 - 4a_\xi c_\xi}}{2a_\xi},
\end{equation}
with $a_\xi  = (k^* - \alpha^2)k^* , b_\xi  =  -2(k^* - \alpha^2)\left(\sum_{k = 1}^{k^*} |v_{[k]}|\right)$ and $c_\xi = \left(\sum_{k = 1}^{k^*} |v_{[k]}|\right)^2 - \alpha^2 \left(\sum_{k = 1}^{k^*} |v_{[k]}|^2\right)$.
After the optimal $t^*$ is obtained, optimal solution for (\ref{beta-update-solution-2}) can be directly calculated by (\ref{beta-t-v}).

The dual coordinate ascent for EPin-sc is summarized in Alg.~\ref{EPin-algorithm-2}. Its output gives an optimal solution for EPin-sc (\ref{EPin-sc-primal}), as guaranteed by Theorem \ref{thm-3}.

\begin{algorithm}[h]
\caption{Dual coordinate ascent for EPin-sc}\label{EPin-algorithm-2}
\Indp   Set $l := 0,\vs^0 := {\bf 0}_{n \times 1}, \vt^0 := -\frac{\tau}{m}{\bf 1}_{m \times 1}$; \\
Calculate $\vw := \sum_{i=1}^m t^0_i y_i \vu_i - \vs^0$;\\
\Repeat{$\vt^{l} = \vt^{l-1}$}
{
\For{ $i = 1, 2, \ldots, m$}
{

\eIf{$c \geq \|\vu_i\|_{2}$}
{
 $d^*_i := \frac{1}{m} - t_i^l$;
}
{
 Calculate $d^*_i$ by (\ref{xi-update-solution});
}

 $\vw := \vw + y_i \vu_i d^*_i$;

 $t^{l+1}_i := t_i^l + d^*_i$;
}
 Set $\vv := \vw + \vs^l$;

 Select $k^*$ satisfying (\ref{k-piecewise}), calculate $\xi^*$ by (\ref{optimal_t}), and $s^{l+1}_i := B_{v_i}(\xi^*)$; \\
 $l: = l + 1$;\\
}
\eIf{$\|\vw\|_2 >0$}
{
$\vx := {\frac{\vw}{\|\vw\|_2}}$;
}
{
Find $\vx$ that satisfies~\eqref{sol_cons_2};
}
\end{algorithm}

\begin{theorem}\label{thm-3}
Alg. \ref{EPin-algorithm-2} converges to an optimum of (\ref{EPin-sc-primal}).
\end{theorem}

\begin{proof}
Denote the output of Alg.~\ref{EPin-algorithm-2} as $\vx^*$ and the corresponding dual variables as $\vs^*, \vt^*$. Then there is
\[
\textstyle   \vs^* = \argmax\limits_\vs ~ -\alpha \|\vs\|_{\infty} - \left\| \sum_{i=1}^m t_i y_i \vu_i - \vs \right\|_{2}.
\]
Suppose $\bar i = \arg \max_i |s^*_i|$, and let $\Delta \vs$ be a vector of which the $\bar i$-th component takes value $\sgn(\vs^*_i)$ and other components equal to zero. The following function
\[
- \alpha \|\vs^* + t \Delta \vs\|_{\infty} - \left\| \vw - t \Delta \vs \right\|_{2}
\]
has the maximal value at $t = 0$.

In the case $\vw \neq {\bf 0}$, $t = 0$ being the maximum of the above function means that
\[
\alpha + \frac{\vw^\top \Delta \vs}{\|\vw\|_{2}} = 0.
\]
Moreover, for any $i: x^*_i \neq 0$, the optimality condition on $s^*_i$ implies that $s^*_i = \|\vs^*\|_{\infty}$. Therefore, we have
\begin{eqnarray*}
(\vs^*)^\top \vx^* & = &\|\vx^*\|_{1} \|\vs^*\|_{\infty}  = \alpha  \|\vs^*\|_{\infty}\\
& \geq & (\vs^*)^{T} \tilde \vx, \quad \forall \|\tilde \vx\|_{1} \leq \alpha.
\end{eqnarray*}
Thus, $\vx^*$ is optimal to (\ref{EPin-sc-primal}).

In the case $\vw = {\bf 0}$, the corresponding dual objective equals to $-\alpha \|\vs^*\|_{\infty} + c \sum_{i=1}^m t^*_i$. Meanwhile, the primal objective is
\begin{eqnarray*}
& & \textstyle \frac{1}{m}\sum_{m=1}^m L_{\tau,c}( - y_i (\vu_i^\top \vx^*)) \\
& = & \textstyle c\sum_{i=1}^m t_i^* - \sum_{i=1}^m t^*_i y_i (\vu_i^\top \vx^*)\\
& = & \textstyle c\sum_{i=1}^m t_i^* - (\vs^*)^{T} \vx^*, \\
& = & \textstyle -\alpha \|\vs^*\|_{\infty} + c \sum_{i=1}^m t^*_i,
\end{eqnarray*}
where the first equality comes from the optimality condition (\ref{sol_cons_2_d})--(\ref{sol_cons_2_f}), the second and the last equality are true because $\vw = {\bf 0}$ and~\eqref{sol_cons_2_c}, respectively. Since the objectives of the primal and dual problems are equal, $\vx^*$ is optimal to (\ref{EPin-sc-primal}).
\end{proof}
Assume that the $\ell_1$-norm of the true signal $\bar \vx$ is known. We set $\alpha=\|\bar\vx\|_1$ for EPin-sc and test its performance for different $\tau$ values in Fig.\ref{fig-SNR-s-sc:a}. Note that $\tau = -1$ corresponds to Plan's model. In many applications, the $\ell_1$-norm of the true signal is not known, and we have to estimate it. We fix $\tau=-0.3$ and show the performance for different $\alpha$ values in Fig.\ref{fig-SNR-s-sc:b}, where $\alpha = \sqrt{K}$ is marked.

\begin{figure}[htbp]
  \centering
  \subfigure[]{
    \label{fig-SNR-s-sc:a} 
    \includegraphics[width=0.42\linewidth]{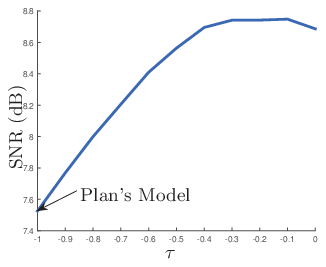}}
  \subfigure[]{
    \label{fig-SNR-s-sc:b} 
    \includegraphics[width=0.42\linewidth]{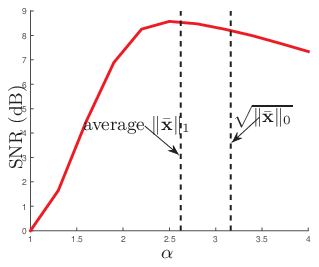}}
  \caption{Average SNR of EPin-sc for different $\tau$ and $\alpha$ values with $m = 500, n = 1000, K = 10, s_n = 10$, and $r_f = 10\%$. (a) We optimally choose $s = \|\bar \vx\|_1$ and test the performance for different $\tau$ values. (b) Set $\tau = -0.3$, and test different $\alpha$ values.
    }\label{fig-SNR-s-sc}
\end{figure}

\section{Numerical Experiments}\label{sec:6}

In the previous sections, we discussed the pinball loss minimization for robust 1bit-CS, proposed two convex models, and designed the corresponding fast algorithms. In this section, we evaluate the performance of the pinball loss minimization in numerical experiments. The data are generated as in the previous section, i.e., we randomly choose $K$ components from a $n$-dimensional signal, draw their values from the Gaussian distribution, and normalize the signal onto the unit $\ell_2$-norm ball. Then, $m$ sign observations are generated by (\ref{noise-data}), where $\varepsilon$ is the Gaussian noise with signal-to-noise ratio $s_n$. We also consider sign flips with ratio $r_f$. The experiments are done with Matlab 2014b on Core i5-3.10GHz and 8.0GB RAM. The source code of the proposed algorithms can be found on the authors' homepage\footnote{\url{http://www.esat.kuleuven.be/stadius/ADB/huang/downloads/1bitCSLab.zip}}. Notice that in our numerical study we only consider Gaussian measurements and normalized signals, which is the mainstream for 1bit-CS. Recently, there are some discussions on non-Gaussian measurements \cite{ai2014one} and the $\ell_2$-norm estimation \cite{knudson2014one}.

Before investigating the recovery quality of the pinball loss minimization, we first evaluate the effectiveness of the dual coordinate ascent algorithms. We compare the computational time of Alg.~\ref{EPin-algorithm} and CVX for solving EPin~\eqref{elastic-net-primal}. More specifically, we vary the number of measurements $m$ from $50$ to $600$, meanwhile keep the ratio $n/m = 2$ and the sparsity level $K/n = 0.02$. The noise level is $s_n = 10, r_f = 10\%$, and the parameter in EPin (\ref{EPin-algorithm}) is chosen as $\mu = \sqrt{{\log(n)}/{m}}, \tau = -0.5$. The average computational time over 200 trials is then reported in Fig.\ref{fig-time:a}, from which we can observe that the proposed dual coordinate ascent algorithm significantly save the computational time from CVX.

\begin{figure}[hptb]
  \centering
  \subfigure[]{    \label{fig-time:a}
    \includegraphics[width=0.45\linewidth]{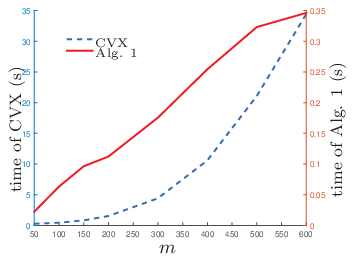}}
  \subfigure[]{
    \label{fig-time:b}
    \includegraphics[width=0.45\linewidth]{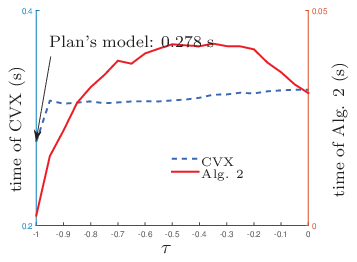}
  }
\caption{Average computational time for EPin by  the proposed algorithms (red solid curve) and CVX (blue dashed curve). (a) $n/m = 2$, $K/n = 0.02$, $\tau = -0.5$, and we vary $m$ from $50$ to $600$. (b) we choose $n = 100, m = 50$, and change $\tau$ from $-1.0$ to $0.0$.}\label{fig-time}
\end{figure}

Similarly, Alg.~\ref{EPin-algorithm-2} can solve EPin-sc (\ref{EPin-sc-primal}) efficiently. Besides the problem size, the computational time is also linked with $\tau$, which controls the feasible set in the dual problem.
In Fig.\ref{fig-time:b}, we report the average computational time of EPin-sc by CVX and Alg.~\ref{EPin-algorithm-2}. The computational time of Plan's model, for which CVX is suggested by \cite{zhang2014efficient}, is marked as well.

Next we consider signal recovery quality for EPin and EPin-sc. Denote the recovered signal as $\tilde \vx$ and the true signal as $\bar x$. The recovery quality then can be measured by SNR (\ref{snr}). Moreover, we will also consider inconsistency ratio $\mathrm{INR}(\bar \vx, \tilde \vx) = {\left |\{i: \sgn(\vu_i^\top \bar \vx) \neq \sgn(\vu_i^\top \tilde \vx)\} \right|}/{m}$.

\begin{figure}[hptb]
  \centering
  \subfigure[]{
    \label{fig-rf:a} 
    \includegraphics[width=0.45\linewidth]{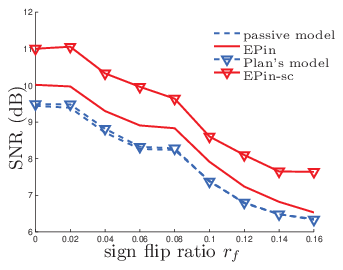}}
  \subfigure[]{
    \label{fig-rf:c} 
    \includegraphics[width=0.45\linewidth]{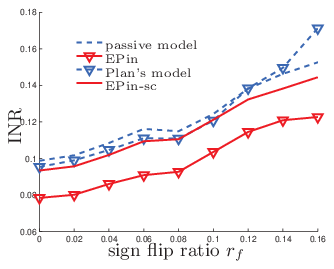}}
  \caption{Recovery performance of the passive model (blue dashed line), EPin (red solid line), Plan's model (blue dashed line with triangle), and EPin-sc (red solid line with triangle) for different sign flip ratios. In this experiment, $n = 1000, m = 500, K = 10$, and $s_n = 10$. The recovery quality is measured by: (a) recovery SNR; (b) inconsistency ratio.}\label{fig-rf}
\end{figure}

As discussed previously, the best choice of $\tau$ and $\mu$ for EPin (\ref{elastic-net-primal}) is problem-dependent and 10-fold cross-validation based on consistency can be used to tune the parameters. Specifically, we randomly partition the data into 10 subsets. In turn, one of the subsets is used for validation and the rest are used for training. For each pair of $\tau$ and $\mu$, we use Alg.~\ref{EPin-algorithm} on the training data and calculate the sign consistency on the validation data. Then $\tau$ and $\mu$ corresponding to the highest consistency are chosen. The parameter candidate set is $\tau \in \{-1, -0.8, -0.6, -0.4, -0.2\}$ and $\mu/\sqrt{{\log(n)}/{m}} \in \{0.6, 0.8, 1.0, 1.2\}$. To make a fair comparison, $\mu$ in passive model (\ref{Zhang}) is also tuned by 10-fold cross-validation. For EPin-sc (\ref{EPin-sc-primal}) and Plan's model (\ref{Plan}), the best $s$ is the $\ell_1$-norm of the true signal. It also can be tuned by cross-validation but in this experiment we set $s = \|\bar \vx\|_1$ to show the best performance of EPin-sc and Plan's model. Additionally, the comparison between EPin with $\mu$ obtained by cross-validation and EPin-sc with an optimal $s$ helps us evaluate the parameter tuning method.

Using the above setting, we evaluate Passive, Plan's model and our proposed algorithms for different sign flip ratios $r_f$. Here,  $n = 1000, m = 500, s_n = 10$ and the average recovery performance of 200 trials for different $r_f$ is displayed in Fig.\ref{fig-rf}. The performance trends for these methods are similar. It is interesting to further consider sign flip detection methods, e.g., adaptive outlier pursuit technique designed in \cite{yan2012robust} for BIHT. These methods have already shown good performance for the one-sided $\ell_1$ loss minimization, and are also promising to improve the performance of the pinball loss minimization.

In the following, we fix  $r_f = 10\%$ but vary $s_n$ from $1$ to $100$ and then report the average performance in Fig.\ref{fig-noise}. From these results, one can observe that the performance is generally stable for different noise levels, showing the robustness of EPin and EPin-sc to Gaussian noise. Moreover, in Fig.\ref{fig-noise:c}, when $s_n \geq 20$, INRs for EPin/EPin-sc are below $0.1$. Notice that in data generation, there are $r_f = 10\%$ sign flips. Then INR being less than $0.1$ implies the tolerance of the sign flips.

\begin{figure}[hptb]
  \centering
  \subfigure[]{
    \label{fig-noise:a} 
    \includegraphics[width=0.48\linewidth]{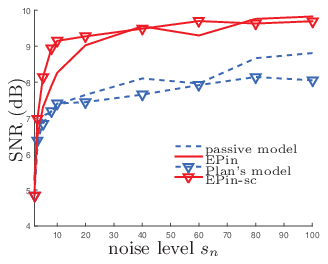}}
  \subfigure[]{
    \psfrag{ins}[c]{\footnotesize INR}
    \psfrag{noise}[c]{\footnotesize noise level $s_n$}
    \psfrag{data1}[l]{\tiny passive model}
    \psfrag{data2}[l]{\tiny EPin}
    \psfrag{data3}[l]{\tiny Plan's model}
    \psfrag{data4}[l]{\tiny EPin-sc}
    \label{fig-noise:c} 
    \includegraphics[width=0.48\linewidth]{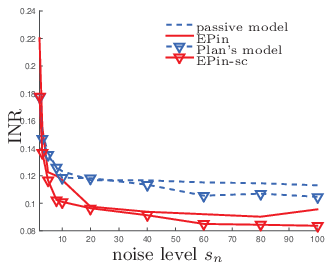}}
  \caption{Recovery performance of the passive model (blue dashed line), EPin (red solid line), Plan's model (blue dashed line with triangle), and EPin-sc (red solid line with triangle) for different noise levels in the case $n = 1000, m = 500, K = 10$ and $r_f = 10\%$. The recovery quality is measured by: (a) recovery SNR;  (b) inconsistency ratio.}\label{fig-noise}
\end{figure}

Next, we fix $s_n = 10$ but increase the number of measurements to compare these methods. The recovered qualities are illustrated in Fig.\ref{fig-m}. With the increasing of $m$, all the recovery quality measures become better. The change trends for different models are similar that EPin and EPin-sc can improve the performance from the existing algorithms. The performance of EPin-sc is slightly better than EPin, which is mainly due to the fact that $s$ is optimally given but $\mu$ is tuned by cross-validation. It indicates that good estimation on the sparsity can help recovering the true signal and in that case EPin-sc is more suitable. Though $\mu$ is chosen based on cross-validation, there is generally no big difference between EPin and EPin-sc with an optimal $s$. If there is no prior-knowledge on $\|\bar \vx\|_0$ or $\|\bar \vx\|_1$, EPin is a good choice and cross-validation on $\mu$ can help.


\begin{figure}[hptb]
  \centering
  \subfigure[]{
    \label{fig-m:a} 
    \includegraphics[width=0.48\linewidth]{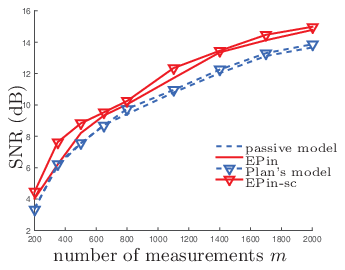}}
  \subfigure[]{
    \label{fig-m:c} 
    \includegraphics[width=0.48\linewidth]{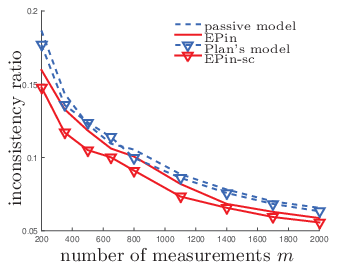}}
  \caption{Recovery performance of the passive model (blue dashed line), EPin (red solid line), Plan's model (blue dashed line with triangle), and EPin-sc (red solid line with triangle) for different number of observations. In this experiment, $n = 1000, K = 10, s_n = 10$, and $r_f = 10\%$. The recovery quality is measured by: (a) recovery SNR; (b) inconsistency ratio.}\label{fig-m}
\end{figure}

{The above observations and comparison keep true for different sparsity levels, as plotted in Fig.\ref{fig-K:a}. When the number of non-zero components is small, the difference among these methods is minor. But when the number is large, Epin and Epin-sc show significant advantage over the other methods, demonstrating the effect of using pinball loss.}

\begin{figure}[hptb]
 \centering
 \subfigure[]{
   \label{fig-K:a} 
   \includegraphics[width=0.48\linewidth]{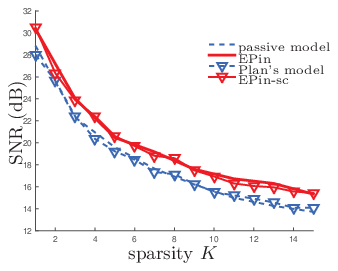}}
 \subfigure[]{
   \label{fig-K:c} 
   \includegraphics[width=0.48\linewidth]{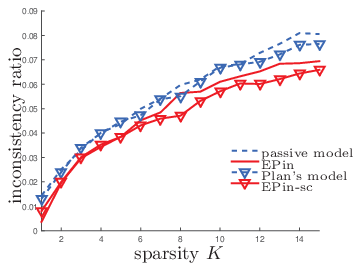}}
 \caption{Recovery performance of the passive model (blue dashed line), EPin (red solid line), Plan's model (blue dashed line with triangle), and EPin-sc (red solid line with triangle) for different sparsity levels in the case $n = 1000, m = 2000, s_n = 10$ and $r_f = 10\%$. The recovery quality is measured by: (a) recovery SNR; (b) inconsistency ratio.}\label{fig-K}
\end{figure}

{At last, we evaluate the proposed EPin and EPin-sc in higher dimensional spaces. In this experiment, we keep the ratio of the number of measurements over the signal dimension as 0.5, and that of the sparsity level over the dimension as 0.01. The SNRs and INRs of reconstructed signal for $n = 1000, 2000, \ldots, 10000$ are given in Fig.\ref{fig-n}. which shows that the reconstruction performance for different dimensional problems is stable for the same sparsity and observations ratio. Note that here we do not tune parameters for different $n$ but use the ones set for $n=1000$. Thus, the performance for $n=1000$ is a bit better than the other cases.}

\begin{figure}[hptb]
  \centering
  \subfigure[]{
    \label{fig-n:a} 
    \includegraphics[width=0.48\linewidth]{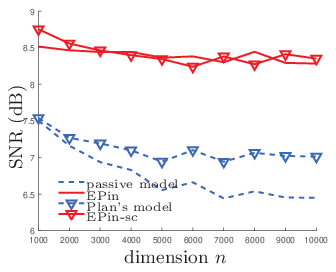}}
  \subfigure[]{
    \label{fig-n:c} 
    \includegraphics[width=0.48\linewidth]{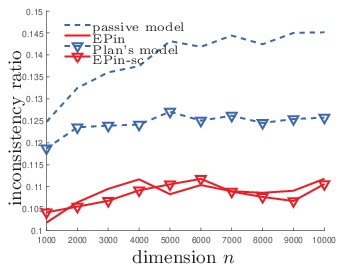}}
  \caption{Recovery performance of the passive model (blue dashed line), EPin (red solid line), Plan's model (blue dashed line with triangle), and EPin-sc (red solid line with triangle) on different dimensions. In this experiment, $m = n/2, K = n/100, s_n = 10$, and $r_f = 10\%$. The recovery quality is measured by: (a) recovery SNR; (b) inconsistency ratio.}\label{fig-n}
\end{figure}

\section{Conclusions}\label{sec:7}

One-bit compressive sensing aims at recovering a signal from a set of sign measurements. Currently, the one-sided $\ell_1$ and the linear loss are the two popular choices for 1bit-CS. Inspired by this observation, a compromise between them, i.e., the pinball loss, is expected to have good recovery performance. In this paper, we analyzed the pinball loss from maximum likelihood and then establish two convex models, EPin and EPin-sc, to deal with 1bit-CS in the presence of noise. The corresponding fast dual coordinate ascent algorithms are proposed, and the convergence is proved. The numerical experiments demonstrate that the proposed algorithms can efficiently find the optimal solutions and the recovery quality is quite good.

{For the future work, the adaptive outlier pursuit corresponding to EPin/EPin-sc is promising to further improve the performance with sign flips. Moreover, replacing the $\ell_1$-norm penalty in EPin/EPin-sc by some non-convex ones could enhance sparsity of the solution and the recovery could be improved, especially when there are not enough observations. The current nonconvex methods for 1bit-CS is mainly for the linear loss \cite{zhu2015towards,huang2018nonconvex,north2015one, chen_onebit_2015}. For other loss functions, e.g., the hinge loss and the pinball loss, the main obstacle is that nonconvex penalties are hard to optimize. The potential techniques include difference of convex functions algorithm and Frank-Wolf algorithm \cite{rinaldi2011concave,rinaldi2010concave,lethi2015concave}.}

\section*{Acknowledgment}
{
\small

This work was partially supported by $\bullet$ EU: The research leading to these results has received funding from the European Research Council under the European Union's Seventh Framework Programme (FP7/2007-2013)/ERC AdG A-DATADRIVE-B (290923). This paper reflects only the authors' views, the Union is not liable for any use that may be made of the contained information. $\bullet$ Research Council KUL: GOA/10/09 MaNet, CoE PFV/10/002 (OPTEC), BIL12/11T; PhD/Postdoc grants. $\bullet$ Flemish Government: FWO: projects: G.0377.12 (Structured systems), G.088114N (Tensor based data similarity); PhD/Postdoc grants IWT: projects: SBO POM (100031); PhD/Postdoc grants iMinds Medical Information Technologies SBO 2014. $\bullet$ Belgian Federal Science Policy Office: IUAP P7/19 (DYSCO, Dynamical systems, control and optimization, 2012-2017). X. H is also supported by Alexander von Humboldt Foundation and the National Natural Science Foundation of China (NSFC 61603248). L. S. is supported by NSFC (11571078), the Joint Research Fund by NSFC and Research Grants Council of Hong Kong (11461161006 and CityU 104012) and ``Zhuo Xue'' program of Fudan University. M. Y. is partially supported by NSF DMS-1621798.
}


%


\end{document}